\documentclass[letterpaper]{article} %
\usepackage[]{aaai2027}  %
\nocopyright
\usepackage[hyphens]{url}  %
\usepackage{graphicx} %
\usepackage{natbib}  %
\usepackage{caption} %
\usepackage{algorithm}
\usepackage{algorithmic}

\usepackage{newfloat}
\usepackage{listings}
\DeclareCaptionStyle{ruled}{labelfont=normalfont,labelsep=colon,strut=off} %
\floatstyle{ruled}
\newfloat{listing}{tb}{lst}{}
\floatname{listing}{Listing}

\usepackage{booktabs}

\usepackage{nicefrac}       %
\usepackage{microtype}      %
\usepackage{xcolor}
\usepackage{pifont}%
\usepackage{tikz}
\usepackage{makecell, multirow, tabularx}
\usepackage{adjustbox}
\usepackage{tcolorbox}
\usetikzlibrary{positioning, calc, arrows.meta, decorations.pathreplacing}
\usetikzlibrary{cd}

\newcommand{\decompiler}[0]{\textsc{Manifold}}
\newcommand{\framework}[0]{Superset Decompilation}
\usepackage{calc}
\usepackage{amsmath, amssymb, amsthm}

\theoremstyle{plain}
\newtheorem{theorem}{Theorem}
\newtheorem{proposition}[theorem]{Proposition}
\newtheorem{lemma}[theorem]{Lemma}

\theoremstyle{definition}
\newtheorem{definition}[theorem]{Definition}

\newcommand{\erase}{\ensuremath{\mathrm{class}}}

\theoremstyle{remark}
\newtheorem{remark}[theorem]{Remark}

\title{Superset Decompilation}

\author{
    Chang Liu\textsuperscript{\rm 1},
    Yihao Sun\textsuperscript{\rm 1},
    Thomas Gilray\textsuperscript{\rm 2},
    Kristopher Micinski\textsuperscript{\rm 1}
}
\affiliations{
    \textsuperscript{\rm 1}Syracuse University\\
    \textsuperscript{\rm 2}Washington State University\\
    cliu57@syr.edu, ysun67@syr.edu, thomas.gilray@wsu.edu, kkmicins@syr.edu
}

\begin{document}

\maketitle

\begin{abstract}

Decompilation is the task of recovering high-level source from a compiled binary. Conventionally, it is implemented as a series of analyses and rewrites over a single representation, committing at each step to a single interpretation and discarding the rest. We argue this premature commitment is why decompilers struggle to produce compilable output on whole real-world programs, which is necessary for further testing. Compilation is many-to-one: a single binary may admit many source programs, and the evidence that resolves an ambiguity often lies far from where it arises. We therefore preserve competing interpretations when their resolution depends on non-local evidence, while resolving purely local ambiguities in a final solving stage. We formalize this approach as Superset Decompilation: monotone Datalog rules accumulate candidate interpretations across an IR hierarchy, and a per-function weighted MaxSMT solver jointly selects a globally type-consistent C program. Our implementation Manifold lifts stripped x86-64 Linux ELF binaries to C99 through CompCert IRs in 50K lines of Rust and Datalog. On GNU coreutils, it achieves 93\% recompilability and passes 7\% of the test suite.

\end{abstract}

\section{Introduction}

Decompilers are conventionally a sequence of imperative transformations over one mutable representation, and they struggle to produce output a compiler accepts on whole real-world programs: on GNU coreutils, Ghidra~\cite{ghidra}, Hex-Rays~\cite{ida}, angr~\cite{angr}, and RetDec~\cite{retdec} all emit C that fails to recompile (Table~\ref{tab:combined}). This is a consequence of the decompiler architecture rather than an engineering shortfall: committing at each step to one interpretation of the binary locally while discarding other possible ones. Such early commitment causes inconsistency and guessed placeholders in the emitted code, so compilers reject the output.

The ambiguity in the decompilation process is structural, inherited from the compilation process: compilation is not injective. Register allocation maps many virtual registers onto one physical location, stack layout commits distinct variables to shared frame offsets, linearization replaces structured control flow with jumps; even disassembling a stripped binary is undecidable in general~\cite{Engel:24, ddisasm}. Decompilation is therefore inference under irreducible ambiguity rather than deterministic inversion: the signedness of a value may be fixed only by a comparison several instructions later, and a loaded value's width may stay consistent with several types no matter how much of the program is examined. Verified decompilation~\cite{Freek2022} removes the ambiguity by inverting a formally specified compiler, but such specifications rarely exist for production compilers or ISAs like x86-64. We take the opposite approach: preserve alternatives whose resolution depends on nonlocal evidence, and commit to those coupled choices only once, globally, in the most coherent representation.

\decompiler{} implements this for x86-64 Linux ELF binaries. The reasoning process is monotone: each step adds candidates and never retracts, while selection resolves the accumulated candidates as one global optimization per function. On coreutils \decompiler{} achieves 93\% recompilability (Table~\ref{tab:combined}); an ablation that swaps the global solve for compiler-guided greedy search over the identical candidate forest recompiles nothing. The closest prior result, DecLLM~\cite{DecLLM}, reports 70\%, but on unstripped binaries, function-by-function, as a post-hoc LLM repair layer over an existing symbolic decompiler; \decompiler{} reaches this by deduction alone, on stripped input, end to end. We claim recompilability, because clearing the compiler does not establish behavioral equivalence: the emitted programs pass only 7\% of the GNU test suite (\S~\ref{sec:test-quality} measures this gap). Even short of behavioral level correctness, the separation is categorical: 13 front-end errors on coreutils against 622 for RetDec and 22k–85k for the rest.

\begin{itemize}

\item We recast decompilation as reasoning under ambiguity,
formalized as \emph{\framework{}}: logic-defined passes monotonically
lift a binary through an IR hierarchy, preserving alternative
interpretations.

\item We reduce program recovery to global optimization: the C type
discipline becomes a weighted constraint theory, and a global optimal program is selected per function by MaxSMT.

\item We present \decompiler{}: ${\sim}$50K lines of Rust and Datalog
lifting stripped x86-64 ELF to C99 through CompCert IRs, whose pass
architecture admits new analyses in isolation.

\item We show \decompiler{} produces \textbf{compilable C on whole,
stripped coreutils binaries}: 93\%
against 0\% for every baseline \emph{and} for a greedy ablation over
the identical forest.

\end{itemize}

\begin{figure*}[t]
\centering
\begin{tikzpicture}[
  font=\small,
  band/.style ={rounded corners=3pt, line width=0.4pt, draw=#1!50, fill=#1!7},
  bandlbl/.style={font=\small\bfseries, anchor=west, inner sep=1pt},
  chip/.style ={rounded corners=2pt, line width=0.4pt, draw=black!55, fill=white,
                inner sep=2pt, minimum height=0.36cm, font=\scriptsize\ttfamily},
  gone/.style ={chip, draw=black!35, dash pattern=on 2pt off 1.6pt, text=black!45},
  lift/.style ={-{Stealth[length=2mm]}, line width=0.5pt, black!55},
  evid/.style ={-{Stealth[length=2mm]}, line width=0.7pt, black!55,
                dash pattern=on 3pt off 2pt},
  tie/.style  ={line width=0.35pt, black!40, densely dotted},
  note/.style ={font=\scriptsize, text=black!60, inner sep=1pt},
]
\node[font=\small\bfseries] at (7.1,6.28)   {Phase 1 $\cdot$ deduction: monotone, PTIME};
\node[font=\small\bfseries] at (14.35,6.28) {Phase 2 $\cdot$ commitment: NP-hard};

\draw[band=black] (3.0,0) rectangle (11.2,0.62);
\node[bandlbl] at (3.1,0.31) {Asm: $D_0$ (EDB)};
\node[chip] at (6.9,0.31) {mov (\%rax), \%eax};
\node[chip] at (9.7,0.31) {mov 0x4(\%rax), \%eax};
\draw[lift] (10.9,0.66) -- (10.9,0.92);

\draw[band=black] (3.0,0.96) rectangle (11.2,1.58);
\node[bandlbl] at (3.1,1.27) {Mach: typed ops};
\node[chip] at (6.9,1.27) {Mload i32, +0};
\node[chip] at (9.7,1.27) {Mload i32, +4};
\draw[lift] (10.9,1.62) -- (10.9,1.88);
\node[note, anchor=west] at (4.6,1.75) {$\cdots$ Linear, LTL \; (addr $\to$ node)};

\draw[band=black] (3.0,1.92) rectangle (11.2,2.82);
\node[bandlbl] at (3.1,2.4) {RTL: many$\to$1};
\node[note, anchor=west] (pat) at (3.1,2.08) {same base $p$ $\Rightarrow$ struct $\{0,4\}$?};
\node[chip] at (6.9,2.44) {Iload +0, [p]};
\node[chip] at (9.7,2.44) {Iload +4, [p]};
\draw[lift] (10.9,2.86) -- (10.9,3.12);
\node[note, anchor=west] at (4.6,3.00) {$\cdots$ Cminor (many)};

\draw[band=black] (3.0,3.16) rectangle (11.2,4.06);
\node[bandlbl] at (3.1,3.61) {Csharpminor: many};
\node[chip] at (9.7,3.80) {Oadd(p,4) 32-bit};
\node[chip] at (9.7,3.42) {Oaddl(p,4) 64-bit};
\draw[lift] (10.9,4.10) -- (10.9,4.36);

\draw[band=black] (3.0,4.40) rectangle (11.2,5.90);
\node[bandlbl] at (3.1,5.66) {Clight: many, global};
\node[chip, anchor=west] (c1) at (8.0,5.58) {v = p->ofs\_4 + 1};
\node[chip, anchor=west] (c2) at (8.0,5.13) {v = p[1] + 1};
\node[chip, anchor=west] (c3) at (7.0,4.68) {v = *(int*)((char*)p+4) + 1};
\node[chip, anchor=west] (tp) at (4.0,5) {$\tau_p$: s1* $|$ int*};
\draw[tie] (tp.east) -- (c1.west);
\draw[tie] (tp.east) -- (c2.west);
\draw[tie] (tp.east) -- (c3.west);

\draw[evid] (pat.west) .. controls (2.5,2.5) and (2.5,4.4) .. (tp.west);
\node[note, anchor=west] at (1.3,3.50) {evidence at};
\node[note, anchor=west] at (1.3,3.22) {a distance};

\draw[band=black] (0.2,1.55) rectangle (2.4,2.50);
\node[font=\small\bfseries] at (1.3,2.20) {original C};
\node[font=\scriptsize\ttfamily] at (1.3,1.85) {classify(p)};
\draw[lift] (1.3,1.51) -- (1.3,1.02);
\node[note, anchor=west] at (1.42,1.26) {gcc -O3};
\draw[band=black] (0.2,0.08) rectangle (2.4,0.98);
\node[font=\small\bfseries] at (1.3,0.68) {stripped ELF};
\node[note] at (1.3,0.35) {x86-64};
\draw[lift] (2.44,0.31) -- (2.94,0.31);

\node[note] at (7.1,-0.42)
  {$D_0 \subseteq D_1 \subseteq \cdots \subseteq D_m$: nothing retracted (Prop.~\ref{prop:monotone})};

\node[note] at (14.35,5.72) {hard: shape kernel};
\node[note] at (14.35,5.42) {soft: legality $\gg$ priors};
\draw[line width=0.4pt, draw=teal!60, fill=teal!10]
  (13,5.15) -- (15.7,5.15) -- (15.35,4.25) -- (13.35,4.25) -- cycle;
\node[font=\small\bfseries] at (14.35,4.85) {MaxSMT};
\node[note] at (14.35,4.53) {(Alg.~\ref{alg:select})};
\draw[lift] (11.28,5.00) -- (12.62,5.00);
\draw[lift] (14.35,4.1) -- (14.35,3.40);
\draw[rounded corners=3pt, line width=0.4pt, draw=teal!60, fill=teal!12]
  (12.55,2.10) rectangle (16.15,3.32);
\node[font=\small\bfseries] at (14.35,3.03) {one C function};
\node[font=\scriptsize\ttfamily] at (14.35,2.70) {int classify(s1 *p)\{};
\node[font=\scriptsize\ttfamily] at (14.35,2.36) {... v = p[1] + 1;\}};
\node[note] at (14.35,1.80) {UNSAT $\Rightarrow$ rule gap (Lem.~\ref{lem:unsat})};
\end{tikzpicture}
\caption{\decompiler{} separates monotone deduction (left) from
per-function MaxSMT commitment (right). Deduction accumulates candidates
without retraction (Prop.~\ref{prop:monotone}); at RTL, many${\to}$1 denotes
per-node resolution of instruction-fusion alternatives from local evidence. The shared type $\tau_p$ carries nonlocal dependencies into
commitment. The solver treats shape as hard and prioritizes legality over
priors (Lemma~\ref{lem:dominance}); UNSAT identifies a generation gap
(Lemma~\ref{lem:unsat}). Pass-through levels are elided; Table~\ref{tab:atoms}
gives the full schema, and \texttt{s1} abbreviates the recovered struct.}
\label{fig:e2e}
\end{figure*}

\section{Background}

\decompiler{} lifts Linux x86-64 ELF binaries to typed C through a
sequence of small, modular passes written in Ascent Datalog~\cite{ascent}
and Rust, following CompCert's IR hierarchy. Each pass records the
readings it derives and preserves alternatives whose resolution depends
on nonlocal evidence; ambiguities supported entirely by local evidence
may be resolved within the pass. Section~\ref{sec:formal} formalizes this
process as computation over an append-only relation store with explicit
derivations for its facts.

\textbf{Decompiler Architecture.}
Established decompilers are monolithic: Hex-Rays~\cite{ida},
Ghidra~\cite{ghidra}, and angr~\cite{angr,sailr} couple their analyses
through one heavily mutable IR, hindering extension and masking
fidelity issues~\cite{effectiveness, dramko2024taxonomy, pcode}, while
lifters like RetDec~\cite{retdec} and McSema~\cite{mcsema} target LLVM,
whose compiler-centric encoding degrades reverse-engineering
precision~\cite{liu2022sok}. \decompiler{} instead imports the
nanopass discipline of modern compiler
construction~\cite{nanopass, llvm}, using
CompCert's~\cite{compcert} stratified IRs as a blueprint for modular
lifting, aiming at extensibility rather than the strict inversion pursued by
verified decompilation~\cite{myreen2012decompilation,Freek2022}. Constraint-based~\cite{Lee2011TIEPR, retypd} and
heuristic~\cite{9519451} type recovery run as
post-hoc passes over a fixed IR; neural methods predict types and
names~\cite{DIRTY, lacomis2019dire, 299717}, translate binaries to
source~\cite{tan2024llm4decompile, hu2024degpt}, or drive agentic
workflows~\cite{DecLLM, zou2025dlift, chen2025recopilot,
basquedecompiling}.

\textbf{Declarative Program Analysis.}
Datalog is a declarative logic language widely used in program analysis and reverse engineering: a program pairs an extensional database (EDB) of base facts with an intensional database (IDB) of facts derived through Horn clauses and iterated to a least fixpoint, with the EDB populated directly from instructions and control-flow edges. Engines such as Souffl\'{e}~\cite{souffle-cav} and Ascent~\cite{ascent} separate relational specification from execution strategy, scaling from pointer analysis in DOOP~\cite{doop} to C++ class recovery~\cite{ooanalyzer}, taint analysis~\cite{ctadl}, and smart-contract decompilation~\cite{decompilationsmartcontracts}. Traditional engines require flat, tabular facts and cannot express the nested structures of compiler IRs (e.g.\ a recursive \texttt{Tpointer(Arc<ClightType>)}); Ascent removes this restriction with a ``Bring Your Own Data Structures'' feature that embeds Rust algebraic types and expressions directly in rule bodies. \decompiler{} uses this feature to \emph{synthesize} nested terms during inference, beyond simply storing them: a lifting rule constructs a fully typed Clight expression from an untyped one by querying the recovered type of each variable, so one source fact can yield several candidate liftings, later reconciled by the global commitment phase (\S~\ref{sec:disambiguation}). Answer Set Programming~\cite{asp} natively offers the choice constructs and weighted soft-constraint optimization our commitment phase needs, but its ground-and-solve evaluation must finitely instantiate every rule before solving, which is impossible for the nested, unbounded typed terms our lifting rules invent, so we keep generation in monotone Datalog, whose least-fixpoint semantics is all that phase requires, and realize the weighted selection by per-function MaxSMT (\S~\ref{sec:selection}). Closest to us is superset disassembly~\cite{Bauman2018SupersetDS}, which decodes every offset and keeps all candidate instructions, resolved by Datalog inference in \textsc{ddisasm}~\cite{ddisasm}, likelihood scoring~\cite{10.1109/ICSE.2019.00121}, or graph aggregation~\cite{Ye2023DARMDA}. These tools apply keep-every-candidate at one level and resolve it on the spot; \framework{} elevates it to the whole pipeline (register allocation, stack layout, and linearization are each many-to-one, so ambiguity recurs at every level), preserving unresolved, globally coupled branches with provenance so that evidence surfacing at higher levels can discharge an ambiguity left open below, and committing to the remaining coupled choices once, globally, at the end (\S~\ref{sec:disambiguation}).

\section{Superset Decompilation}
\label{sec:formal}

\framework{} resolves the non-injectivity of compilation by treating decompilation as structured search through the space of higher-level explanations of a binary: a poset of progressively higher-level IRs, with relational passes lifting facts upward and each pass only adding knowledge and auxiliary analysis facts at the next level. We assume no formal specification of the compiler or its inverse image: such specifications rarely exist for production compilers or for x86-64, and real binaries contain code outside any single compiler's image (e.g., link-time optimization). \framework{} is thus complementary to verified decompilation~\cite{Freek2022}: rather than proving inversion of one compiler, it accumulates a forest of candidate decompilations, each with explicit provenance, in an increasing sequence of relation stores.

Figure~\ref{fig:e2e} traces a running example. \texttt{classify} takes a \texttt{Point* p} (\texttt{struct Point\{int x; int y;\}}) and returns \texttt{p->y+1} if \texttt{p->x} is positive, else \texttt{p->y-1}; GCC lowers \texttt{p->y} to a load at byte offset 4, and allocation yields \texttt{mov 0x4(\%rax), \%eax}. Disassembly (Capstone or \texttt{ddisasm}) populates the store with instructions, symbols, and ABI facts; lifting then ascends CompCert's hierarchy~\cite{compcert} from its near-assembly IRs. At Mach, instructions become typed operations: \texttt{mov (\%rax), \%eax} $\mapsto$ \texttt{Mload(MInt32, Aindexed(0), [AX], AX)}. Linear/LTL replace raw stack offsets with typed slots and normalize SP-relative to BP-based addressing. RTL reverses register allocation, hardware registers becoming pseudo-registers: \texttt{Iload(MInt32, Aindexed(4), [p], y\_var)}. Csharpminor preserves both width readings while making the offset explicit pointer arithmetic, for example \texttt{Eload(MInt32, Ebinop(Oaddl, Evar(p), Econst(4)))}. At Clight, type inference places 32-bit integers at offsets 0 and 4 of \texttt{p}, and struct recovery rewrites the arithmetic as \texttt{p->ofs\_4}. The example fixes the paper's two central ideas. \emph{Evidence at a distance:} the \texttt{p->x} load's signedness is settled only downstream, where the value flows through \texttt{test \%eax, \%eax} into a signed branch (\texttt{jle}). \emph{Irreducibly global ambiguity:} identically typed values at offsets 0 and 4 leave both struct and array readings open, so both travel as parallel candidates into the global solve; an overly conservative pipeline instead flattens \texttt{p->ofs\_4} to \texttt{*(int*)((char*)p+4)}, correct but structureless.

\subsection{Semantic Domains for \framework{}}
\label{sec:decompilation-domain}

\begin{definition}[IR Hierarchy]
\label{def:ir-hierarchy}
An \emph{IR hierarchy} is a finite partially ordered set
$(\mathcal{L}, \preceq)$ whose elements are intermediate
representations. The order relation $\preceq$ is interpreted so that
$\ell \preceq \ell'$ means that $\ell'$ is at least as high-level as
$\ell$, equivalently, that facts at level $\ell'$ may be derived from
facts at level $\ell$ by some sequence of passes.
\end{definition}
In our reversal of CompCert, the IR dependency graph is the totally-ordered chain mirroring CompCert's compilation stages:
\texttt{Asm} $<$ \texttt{Mach} $<$ \texttt{Linear} $<$ \texttt{LTL} $<$
\texttt{RTL} $<$ \texttt{Cminor} $<$ \texttt{Csharpminor} $<$
\texttt{Clight}.  This chain is a
\emph{reference hierarchy of representations}: it organizes the forms of
facts produced by \decompiler{}, but does not imply that each pass is a
formally specified inverse of the corresponding CompCert pass.
Higher-level candidates need not correspond
one-to-one with concrete instruction addresses.  We therefore separate
the identity of a program point from its representation at a particular level.

\begin{definition}[Nodes]
\label{def:nodes}
Let $N = \mathit{Addr} \cup \mathit{SynAddr}$
be the set of \emph{nodes}.  Elements of $\mathit{Addr}$ are concrete
machine addresses.  Elements of $\mathit{SynAddr}$ are synthetic nodes
introduced during lifting, such as control-flow join points or
regions lifted from no single address.
\end{definition}

\begin{definition}[Relation Store]
\label{def:store}
For each IR level $\ell\in\mathcal{L}$, let
$\mathcal{R}_\ell=\{r_1{:}\tau_1,\ldots,r_{k_\ell}{:}\tau_{k_\ell}\}$ be a
finite relation schema, and let
$\mathcal{R}=\bigcup_{\ell}\mathcal{R}_\ell$ together with the auxiliary
analysis relations. A \emph{relation store} over $\mathcal{R}$ maps each
relation $r{:}\tau$ to a finite set
$D(r)\subseteq\operatorname{Tuples}(\tau)$. We write $D\subseteq D'$ when
$D(r)\subseteq D'(r)$ for every $r$.
\end{definition}

\begin{definition}[Statement Candidates]
\label{def:candidates}
For each IR level $\ell$ and node $n\in N$,
\[
\mathit{Cand}_\ell(n)\;=\;\{\,s\in S_\ell \mid (n,s)\in D(r_\ell)\,\},
\]
where $S_\ell$ is the set of statements at level $\ell$ and $r_\ell$ the
principal statement relation for that level.
\end{definition}

\subsection{Deduction: Passes, Pipeline, and Coverage}
\label{sec:passes}

We model the decompiler as a sequence of \emph{passes} deriving new facts
into the shared store. Passes differ only in \emph{which} relations they
derive (auxiliary analysis facts or principal candidates at a higher IR
level) but are semantically the same kind of object: stratified Datalog
programs. Negation appears in our rules only as \emph{guards}
(e.g., suppressing a generic lifting when a more specific analysis fact is
present), and only over relations that are already closed.

\begin{definition}[Pass]
\label{def:pass}
A \emph{pass} is a tuple $P=(\mathcal{I},\mathcal{O},\Delta)$ where
$\mathcal{I},\mathcal{O}\subseteq\mathcal{R}$ and $\Delta$ is a finite set
of Datalog rules whose heads use relations in $\mathcal{O}$ and whose body
literals are positive or negated atoms over $\mathcal{I}$. A pipeline is
\emph{stratified} if every relation occurring negated in some rule is
\emph{closed} at the point that rule is evaluated: it is derived only by
earlier passes (or earlier strata of the same pass) and never extended
afterwards.
\end{definition}

\begin{definition}[Immediate Consequence Operator]
\label{def:tp}
For a stratum with rules $\Delta_s$, the operator $T_s$ extends the store
with every fact $\theta(\mathrm{head}(\rho))$ such that $\rho\in\Delta_s$
and the substitution $\theta$ satisfies $\rho$'s body: $\theta(a)\in D$
for each positive atom $a$ and $\theta(a)\notin D$ for each negated atom.
A pass evaluates its strata in order, each to a fixpoint.
\end{definition}

\begin{proposition}[Stratum Monotonicity and Append-Only Evolution]
\label{prop:monotone}
Within a stratum, all negated relations are closed, so $T_s$ is monotone
in the stratum's output relations and reaches a least fixpoint after
finitely many iterations. Across strata and passes the store only grows:
$D_0\subseteq D_1\subseteq\cdots\subseteq D_m$. No fact, once derived, is
ever retracted.
\end{proposition}

Note what is \emph{not} claimed: the pipeline is not monotone in the
closed relations themselves, since adding a fact to a guard relation can
suppress a candidate a later stratum would otherwise derive. This is
deliberate (guards implement priority among rules), and it is one reason
completeness below is conditional.

\begin{definition}[Implemented Derivation Step]
\label{def:step}
An \emph{implemented derivation step} is a rule instance whose positive
body facts are present in the store and whose negated atoms are absent
when its stratum is evaluated (equivalently, absent from $D_m$, since
negated relations are closed). A \emph{derivation} for a fact $(r,t)$ is
a finite tree of steps rooted at $(r,t)$ whose leaves are extensional
facts of $D_0$; each step carries its negative side-conditions.
\end{definition}

\begin{theorem}[Coverage]
\label{thm:coverage}
If a fact $(r,t)$ admits a derivation (Definition~\ref{def:step}) whose
leaves lie in $D_0$ and whose negative side-conditions hold in $D_m$, then
$(r,t)\in D_m$; in particular such a statement candidate $(n,s)$ satisfies
$s\in\mathit{Cand}_\ell(n)$.
\end{theorem}
\begin{proof}
Induction on derivation height. Leaves are extensional. For a step
$h\leftarrow b_1,\ldots,b_k,\neg c_1,\ldots,\neg c_j$: each $b_i$ is in the
store when the step's stratum runs (induction hypothesis and
Proposition~\ref{prop:monotone}); each $c_i$ is absent from $D_m$, hence,
by append-only evolution, from every earlier store. The rule fires, and
Proposition~\ref{prop:monotone} preserves $h$ to $D_m$.
\end{proof}

\subsection{Commitment}
\label{sec:selection}

No monotone fixpoint can play selection's role: saturation derives every
candidate the evidence supports, whereas selection must keep exactly one
per node. Selection's input is the residual ambiguity of
Table~\ref{tab:atoms}'s two global rows: per node $n$ the Clight
candidates $\mathit{Cand}_{\texttt{Clight}}(n)$, per inferable register
$r$ a provenance-ranked list of concrete C types
$T_r=\langle t_{r,1},t_{r,2},\ldots\rangle$. An \emph{assembly} picks
one candidate per node plus a type assignment; the two are coupled (the
type of $r$ decides which candidates are well-typed wherever $r$ is
read), so commitment is joint, once per function.

\begin{definition}[Type-class algebra]
\label{def:classalg}
Let $\mathbb{C}$ be the C types and $\mathbb{T}$ the free datatype
$t ::= \texttt{void} \mid \texttt{int} \mid \texttt{float} \mid
\texttt{func} \mid \texttt{struct}\,\sigma \mid \texttt{union}\,\sigma
\mid \texttt{ptr}\,t$ ($\sigma$ over composite identifiers). The
\emph{class erasure} $\erase:\mathbb{C}\to\mathbb{T}$ forgets width and
signedness and acts homomorphically on pointers and composites;
$\mathbb{T}$ is $\mathbb{C}$'s quotient under $\erase$-equality.
\end{definition}

Each candidate compiles to typing \emph{obligations} over
$\mathbb{T}$: the $\erase$-image of the transcribed C-frontend
judgments (\texttt{Ctypes}, \texttt{Cop}, \texttt{Ctyping}), held in
lockstep with an executable checker by parity tests
(\S~\ref{sec:impl}). The quotient keeps the joint solve tractable:
legality is width-insensitive, so width and signedness carry no
constraint information and are restored after solving by anchoring
against $T_r$ (\S~\ref{sec:disambiguation}).

\begin{definition}[Selection instance]
\label{def:instance}
For a function $f$: variables $\tau_r\in\mathbb{T}$ per register and
Booleans $x_{n,i}$ per node with candidates $s_{n,1},\ldots,s_{n,m_n}$,
under $\sum_i x_{n,i}=1$ ($x_{n,1}:=\top$ when $m_n=1$). The
\emph{hard kernel} $\mathcal{H}$ contains only (i)~totality;
(ii)~domain facts pruning $\tau_r$ (never \texttt{void},
\texttt{func}, \texttt{union}; \texttt{float} only given float
evidence); (iii)~guarded \emph{shape} obligations:
$x_{n,i}\Rightarrow\mathrm{is\_ptr}(\tau_r)$ when $s_{n,i}$
dereferences $r$, $x_{n,i}\Rightarrow\mathrm{is\_struct}(\tau_r)$ when
it selects a member. All else is soft in three tiers: priors
$\tau_r=\erase(t_{r,1})$ and $x_{n,1}$ at weight~$1$;
coercion-economy clauses against gratuitous or lossy casts at $w_c$;
every \emph{legality} obligation (convertibility at assignments,
arguments, returns; operator typing; condition, switch, and callee
classification) at $w_{\top}$, each under its candidate's guard. A
\emph{selection} is a model of $\mathcal{H}$; \emph{optimal} if it
minimizes violated soft weight (MaxSMT).
\end{definition}

Legality is soft by design: C's coercions make nearly every violation
repairable by an explicit cast, so a violated clause costs an inserted
cast or a counted diagnostic rather than voiding the selection, giving
totality relative to the shape kernel at a measured quality cost. Two
properties of the weights carry the recompilability argument.

\begin{lemma}[Evidence dominance]
\label{lem:dominance}
Let $W$ total the non-legality soft weight and $w_{\top}>W$. A selection is optimal iff it
minimizes violated legality clauses and, among such, violated
non-legality weight. If any selection satisfies all legality
obligations, every optimal one does.
\end{lemma}
\begin{proof}
With $0\le p(\mu)\le W<w_{\top}$,
$\mathrm{cost}=w_{\top}L+p$ orders selections lexicographically by
$(L,p)$: one fewer legality violation saves $w_{\top}$, more than any
attainable difference in $p$.
\end{proof}

Evidence decides; priors break ties. Whenever generation produces a
fully legal reading, the solver returns one
(\S~\ref{sec:evaluation}). Infeasibility has an equally sharp meaning:
$\mathbb{T}$ is free, so distinct constructors are disjoint, and two
sole-candidate nodes can force
$\mathrm{is\_ptr}(\tau_r)\wedge\mathrm{is\_struct}(\tau_r)$, so the
instance is \emph{not} always satisfiable.

\begin{lemma}[Infeasibility certifies generation gaps]
\label{lem:unsat}
The instance is satisfiable iff the forest contains a
shape-consistent assembly under the domain facts. UNSAT therefore
certifies that generation produced no shape-consistent reading (a rule
gap), never a failure of ambiguity resolution.
\end{lemma}
\begin{proof}
Feasibility depends on $\mathcal{H}$ alone. A shape-consistent
assembly extends to a model by typing unconstrained registers
\texttt{int}, which every domain admits; conversely, any model
restricts to its choice function.
\end{proof}

\begin{proposition}[Conditional decoupling]
\label{prop:decouple}
Fix $\tau:=M(\tau)$ for any $M\models\mathcal{H}$. Every clause then
mentions at most one node's selection variables (guards bind single
candidates; obligations mention only $\tau$ and constants), so the
objective decomposes as $\sum_n\mathrm{cost}_n(x_n)$: jointly optimal
iff node-wise optimal. This licenses Algorithm~\ref{alg:select}'s
refinement phases: with $\tau$ fixed, each node is re-scored in
isolation against the \emph{exact} frontend judgment, at no solver
cost and no risk of cross-node inconsistency.
\end{proposition}

\begin{remark}[Two complexity regimes]
\label{rem:hardness}
Selection is NP-hard because weighted \textsc{MaxSat} embeds into it,
using one two-class register per variable and one candidate per clause
literal. Generation, by contrast, retains Datalog's polynomial data
complexity.
\end{remark}

\section{Implementation}
\label{sec:impl}

\begin{table}[t]
\centering
\setlength{\tabcolsep}{4pt}
\small
\begin{tabular}{@{}lllll@{}}
\toprule
\textbf{Level} & \textbf{Atom} & \textbf{Key} & \textbf{Cand.} & \textbf{Resolved by} \\
\midrule
Asm    & \texttt{asm\_inst}    & addr & 1      & EDB (disasm.) \\
Mach   & \texttt{mach\_inst}   & addr & many & carried fwd. \\
Linear & \texttt{linear\_inst} & addr & many   & carried fwd. \\
LTL    & \texttt{ltl\_inst}    & node & many   & carried fwd. \\
RTL    & \texttt{rtl\_inst}    & node & many${\to}$1 & local: scored \\
Cminor & \texttt{cminor\_stmt} & node & many   & carried fwd. \\
Csh.   & \texttt{csh\_stmt}    & node & many   & carried fwd. \\
Clight & \texttt{clight\_stmt} & node & many   & {global: MaxSMT} \\
Types  & \texttt{var\_type}    & reg. & many   & {global: MaxSMT} \\
\bottomrule
\end{tabular}
\caption{Store schema by level. The atom $r_\ell(n,s)$ reads ``node $n$
admits reading $s$ at level $\ell$''; edge relations elided.
\emph{Cand.}: readings per key (Def.~\ref{def:candidates}); many${\to}$1 denotes
multiple readings locally scored down to one; the two \emph{global}
rows are the axes of the per-function solve (Alg.~\ref{alg:select},
Def.~\ref{def:instance}).}
\label{tab:atoms}
\end{table}

Table~\ref{tab:atoms} gives the store's schema. A node is a
\texttt{u64}, with low bits the instruction address and reserved tag
bits for synthetic nodes (Definition~\ref{def:nodes}), and node
identity is shared across levels, so evidence at one level constrains
candidates at another: the signed branch of \S~\ref{sec:formal}'s
example acts on the very node whose load it disambiguates. Statements
are full algebraic terms (Ascent's BYOD), so nesting lives
\emph{inside} the atom; control flow lives \emph{between} atoms as an
explicit edge relation, with CompCert's embedded successor fields
stripped. The last two columns are the commitment discipline:
deterministic liftings carry one candidate; ambiguity whose deciding
evidence is local (instruction-fusion idioms at RTL) is resolved per node
on the spot; ambiguity coupled to the type assignment---statement selection,
including the operand-width alternatives surfaced at Csharpminor---reaches
the global solve intact.
Proposition~\ref{prop:decouple} certifies the discipline: local
resolution is globally harmless exactly where it is independent of the
type assignment.

\subsection{Superset Representation}
\label{sec:irdesign}

In forward compilation, each IR carries invariants established by construction: \texttt{RTL} assumes single-definition pseudo-registers, \texttt{Mach} a compiler-chosen stack layout, and \texttt{Clight} a fully typed AST. Raw binaries satisfy none of these. \decompiler{} borrows CompCert's statement constructors (\texttt{Msetstack}, \texttt{Iop}, \texttt{Sassign}, \texttt{Sset}) but stores them as candidate tuples in a central relation store, \texttt{DecompileDB}, which maps string-keyed relation names to type-erased vectors of tuples. Type erasure decouples passes: each needs only agree on a relation name and tuple type, and the scheduler can compose or parallelize them freely. Ascent-based decompile passes \emph{swap} entire relations out of the store, run Datalog to fixpoint, and swap the enlarged relations back without copying; analysis passes instead accumulate facts incrementally.

Each IR level defines a principal statement relation (e.g., \texttt{mach\_inst} or \texttt{clight\_stmt}) and a separate edge relation (\texttt{clight\_succ}). A single node may map to multiple statement tuples, encoding the candidate set of Definition~\ref{def:candidates}. The store also hosts auxiliary analysis relations: register mappings, stack slot descriptors, type-evidence candidates, struct layout hypotheses, and function signatures. All relations (statements, edges, and analysis facts) are subject to the same append-only discipline formalized in Proposition~\ref{prop:monotone}.

\subsection{The Pass Library}
\label{sec:analysispass}

The library instantiates \S~\ref{sec:passes}: decompile passes lift the principal relations level to level, analysis passes derive the auxiliary facts they consume. Lifting begins by recognizing x86-64 mnemonics into typed Mach constructors and normalizing stack accesses into typed slot descriptors (\texttt{Local}, \texttt{Incoming}, \texttt{Outgoing}), marking callee-save spills for suppression, before reorganizing the instruction stream into a CFG. Reversing register allocation takes the most machinery: a union-find pass groups physical-register uses into def-use equivalence classes, each receiving a fresh pseudo-register, and LTL lifts to RTL by substitution. Where x86 deviates from CompCert's conventions (\texttt{IDIV} writes quotient \emph{and} remainder), dedicated rules emit distinct outputs, with synthetic nodes preserving the one-statement-per-node invariant. RTL operations then restructure into expression trees by arity dispatch (\texttt{Econst}, \texttt{Evar}, \texttt{Eunop}, \texttt{Ebinop}); a structuring pass uses dominator trees to emit alternative loop, conditional, and switch-chain structurings as parallel Cminor candidates; and a Csharpminor pass makes addressing explicit pointer arithmetic. The final pass combines the preserved Csharpminor readings with the feasible type assignments, each pairing yielding a distinct Clight candidate, the globally coupled residue the solver receives (Table~\ref{tab:atoms}).

Analysis passes supply the facts this lifting consumes: stack-frame def-use analysis through stack-pointer adjustments; an RTL cleanup fixpoint of copy propagation, dead-store elimination, and live-range merging; type inference combining opcode evidence, pointer-evidence accumulation, and signature-seeded constraint propagation; struct recovery from multi-offset dereference patterns; and signature reconciliation merging definition and call site evidence. The library also extends past CompCert's image without touching existing passes: CompCert lacks VLAs, so one added pass recognizes the allocation idiom (a dynamic \texttt{SUB} of \texttt{RSP} by a register operand, then capture of the adjusted pointer) and emits \texttt{Mbuiltin(alloca, size\_reg)} into the Mach store, which the existing built-in machinery carries unchanged to a Clight \texttt{alloca} call.

\subsection{Selection and Refinement}
\label{sec:disambiguation}

Section~\ref{sec:selection} defines the MaxSMT encoding and its
guarantees. The implementation runs one Z3 solve~\cite{z3} per function in
parallel, then refines its class-level result against the executable
frontend checker (Algorithm~\ref{alg:select}). Parity tests keep the
checker and the transcribed CompCert judgments in agreement.

\begin{algorithm}[t]
\caption{Per-function commitment}
\label{alg:select}
\begin{algorithmic}[1]
\REQUIRE per node $n$: ranked statements $s_{n,1..m_n}$; per register $r$: ranked types $T_r{=}\langle t_{r,1},\dots\rangle$
\STATE declare $\tau_r\in\mathbb{T}$ per register; $x_{n,i}\in\mathbb{B}$, $\sum_i x_{n,i}{=}1$, per node with $m_n{>}1$
\FORALL{nodes $n$, candidates $s_{n,i}$}
\STATE emit each obligation $\varphi$ of $s_{n,i}$ as $x_{n,i}{\Rightarrow}\varphi$: shape into $\mathcal{H}$ (hard); coercion-economy into $\mathcal{S}$ (soft, $w_c$); legality into $\mathcal{S}$ (soft, $w_{\top}$)
\ENDFOR
\STATE add priors $\tau_r{=}\mathrm{class}(t_{r,1})$ and $x_{n,1}$ to $\mathcal{S}$ at weight $1$
\STATE $M\gets\operatorname*{arg\,min}_{\mu\,\models\,\mathcal{H}}\ \textstyle\sum_{(\varphi,w)\in\mathcal{S},\ \mu\not\models\varphi}\,w$ \COMMENT{MaxSMT}
\FORALL{nodes $n$ with $m_n>1$}
\STATE $i^{\star}_{n}\gets\operatorname{lexmin}_{i}\ (e_{n,i},\,h_{n,i},\,c_{n,i},\,v_{n,i},\,i)$ \COMMENT{re-pick}
\ENDFOR
\FORALL{registers $r$}
\STATE $\mathrm{decl}(r)\gets\mathrm{anchor}(T_r,\,M(\tau_r))$
\ENDFOR
\end{algorithmic}
\end{algorithm}

Once $M$ fixes the type classes, Proposition~\ref{prop:decouple}
makes two local refinements safe. First, each multi-candidate node is
re-scored against the executable checker, lexicographically by exact
frontend-error count $e_{n,i}$ under the declarations induced by $M$,
violated structural obligations $h_{n,i}$, lossy $64{\to}32$
truncations $c_{n,i}$, violated soft weight $v_{n,i}$, and provenance
rank $i$. Second, each class-level $M(\tau_r)$ is anchored to $r$'s
candidate list: an exact match when available, otherwise the nearest
width in the same class. This restores width and signedness where the
class encoding is indifferent. After the selected trees are assembled,
a final program-level vote over all use sites settles declarations that
no single-function solve can see, such as a data global used as a
function.

\section{Evaluation}
\label{sec:evaluation}

We compare \decompiler{} with existing decompilers by following the same path a user would take with the recovered C: compile it, then run the available tests. We report how often the output recompiles and how many compiler front-end errors remain. Compilation is a prerequisite for executing the recovered program, but is not evidence that the program behaves correctly.

\subsection{Experiment Setup}

The whole-program benchmarks are GNU Coreutils~\cite{coreutils} 9.10,
compiled for x86-64 Linux with GCC 11.4.0 at \texttt{-O3}, and
Assemblage~\cite{liu2024assemblageautomaticbinarydataset} and
Assemblage-DeepHistory~\cite{liu2026assemblagedeephistorycrossbuildbinarydataset}.
For each Assemblage corpus, we use a stripped subset of 1,000 binaries
balanced across optimization levels. At function level, we use
Decompile-Bench~\cite{decompile-bench} together with the re-executable
HumanEval~\cite{humaneval} and ExeBench~\cite{exebench} benchmarks.
All binaries are dynamically linked and stripped of symbol and
relocation information. From coreutils we evaluate 101 programs,
excluding utilities whose implementations share a single source file
(e.g., \texttt{base64} and \texttt{base32}, both built from
\texttt{basenc.c}), because automated source-to-decompilation
comparison is infeasible for them.

Our classical baselines are Hex-Rays~8.3~\cite{ida},
Ghidra~12.0.3~\cite{ghidra}, angr~9.2.162~\cite{angr}, and
RetDec~5.0~\cite{retdec}, each with a 10-minute timeout per binary. We
also evaluate AutoDecompiler~\cite{autodecompiler},
LLM4Decompile~\cite{tan2024llm4decompile},
SK2Decompile~\cite{sk2decompile}, and
Idioms~\cite{dramko2025idiomsneuraldecompilationjoint}. We omit
DecLLM~\cite{DecLLM} because its source code is not available. All
experiments run on a server with an AMD EPYC 7713P (128 logical
cores) and 500 GB of RAM under Ubuntu 22.04. For LLM-based decompilation,
we serve the models on a four-GPU NVIDIA A40 cluster and follow each
model's original experimental setup.

\begin{table*}[t]
\centering
\setlength{\tabcolsep}{4pt}
\begin{tabular}{@{}l rrr rr r rr r@{}}
\toprule
& \multicolumn{5}{c}{\textbf{Binary \& Libraries}}
& \multicolumn{4}{c}{\textbf{Function / Fragment Level}} \\
\cmidrule(lr){2-6} \cmidrule(lr){7-10}
& \multicolumn{3}{c}{coreutils}
& \makecell{Assemblage-\\Linux}
& \makecell{Assemblage-\\DeepHistory}
& \makecell{Decompile-\\Bench}
& \multicolumn{2}{c}{HumanEval}
& ExeBench \\
\cmidrule(lr){2-4} \cmidrule(lr){5-5} \cmidrule(lr){6-6}
\cmidrule(lr){7-7} \cmidrule(lr){8-9} \cmidrule(lr){10-10}
& recomp. & errors & tests
& recomp. & recomp.
& recomp. & recomp. & re-exec & recomp. \\
\midrule
angr           & 0.00 & 85\,461  & 0.00 & 0.00 & 0.00 & 0.04 & 0.00 & 0.00 & 0.00 \\
Ghidra         & 0.00 & 51\,100  & 0.00 & 0.00 & 0.00 & 0.00 & 0.56 & 0.50 & 0.96 \\
Hex-Rays       & 0.00 & 24\,348  & 0.00 & 0.00 & 0.00 & 0.08 & 0.69 & 0.63 & 0.85 \\
RetDec         & 0.00 &     622  & 0.00 & \textbf{0.30} & 0.11 & \textbf{1.00} & 0.62 & 0.35 & \textbf{0.99} \\
\midrule
LLM4Decompile  & ---   & ---       & ---   & ---   & ---   & 0.12 & 0.79 & 0.53 & 0.61 \\
AutoDecompiler & 0.00   & 22\,003       & 0.00   & 0.00   & 0.00   & 0.09 & 0.39 & 0.23 & 0.63 \\
SK2Decompile   & ---   & ---       & ---   & ---   & ---   & 0.47 & 0.97 & \textbf{0.74} & 0.92 \\
Idioms         & ---   & ---       & ---   & ---   & ---   & 0.18 & 0.79 & 0.34 & 0.77 \\
\midrule
\decompiler{}       & \textbf{0.93} & \textbf{13} & \textbf{0.07} & 0.27 & \textbf{0.32} & 0.96 & \textbf{1.00} & 0.53 & \textbf{0.99} \\
\bottomrule
\end{tabular}%
\caption{Decompilation results across whole-program and function-level
benchmarks. Neural LLM decompilers other than AutoDecompiler are single-function tools that do not
operate on whole binaries. Best value per column in \textbf{bold}.
\emph{recomp.}: recompilable;
\emph{errors}: compiler front-end error count (lower is better);
\emph{tests}: fraction of the GNU coreutils test suite passed;
\emph{re-exec}: re-executable under the original evaluation protocol.}
\label{tab:combined}
\end{table*}

\subsection{Compilation}
\label{sec:compilation-quality}

The clearest difference appears on coreutils. \decompiler{}'s output can be compiled and linked on 93\% of the 101 programs, producing runnable executable, whereas none of the outputs from Ghidra,
Hex-Rays, angr, or RetDec recompiles (Table~\ref{tab:combined}). This
is not a near tie hidden by a binary metric: \decompiler{} produces 13 front-end errors across the
corpus, compared with 622 for RetDec and 22k–85k for the
other decompilers.

This is the failure mode that global selection is intended to prevent.
One locally plausible but globally inconsistent choice can
prevent an entire program from compiling; \decompiler{} instead resolves
coupled choices in one per-function solve. When the candidate forest
contains a fully legal reading, Lemma~\ref{lem:dominance} guarantees
that the solver returns one, so the remaining failures reflect
generation producing no legal reading
(Lemma~\ref{lem:unsat}), rather than a failure of the selector.

The other corpora are more mixed. On
Assemblage-Linux, RetDec recompiles 30\% and narrowly
exceeds \decompiler{}'s 27\%; on Assemblage-DeepHistory,
\decompiler{} leads with 32\% to RetDec's 11\%. At
function level the gap largely disappears: RetDec recompiles all of
Decompile-Bench against \decompiler{}'s 96\%, while
\decompiler{} recompiles all of HumanEval against SK2Decompile's 97\%
and ties RetDec on ExeBench at 99\%.

To isolate the effect of global selection, we replace
Algorithm~\ref{alg:select} with our initial compiler-in-the-loop greedy
search: it seeds each function with the top-ranked candidate per node,
parses Clang
diagnostics to steer a directed fix, and
accepts a single-node replacement only if it strictly lowers the total
error count. Over the identical candidate forest, this selector
recompiles 0\% of coreutils and leaves 60K front-end errors. The failure
is structural: candidates are coupled through shared type variables,
and Proposition~\ref{prop:decouple} decouples nodes only after
$\tau$ is fixed. A repair at one use site can therefore introduce
errors at other uses of the same register; error count is non-monotone
along paths to a legal reading, and strict-descent search stalls in
local optima. The superset representation is necessary
but not sufficient, and deferral pays only when resolution is joint.

\subsection{Execution and Tests}
\label{sec:test-quality}

Compilation of output is an important step in decompilation. Although 93\% of the coreutils
programs recovered by \decompiler{} compile, the resulting programs
pass only 7\% of the coreutils integration tests. The other decompilers pass no
coreutils tests because none of the evaluated whole-program outputs
compile. \decompiler{} therefore makes whole-program behavioral testing
possible, but its low pass rate also leaves a substantial correctness
gap.

At function level, where most tools produce compilable output,
\decompiler{} does not have the best behavioral result. It passes 53\%
of the HumanEval execution tests, compared with 63\% for Hex-Rays and
74\% for SK2Decompile (Table~\ref{tab:combined}). These numbers are not
directly comparable to the whole-program result. The neural refiners
start from symbol-rich pseudocode produced by Hex-Rays or IDA rather
than from stripped bytes, and a function-level test succeeds in
isolation. A coreutils test, by contrast, requires every function on
the executed path to be correct. The function-level benchmarks are
useful measures of local behavioral fidelity, but they do not capture
the compounding errors of a recovered program.

The aggregate results also hide individual cases in which
\decompiler{} succeeds for reasons tied to its architecture. We inspect
the 13 HumanEval functions for which only \decompiler{} passes the
tests, consistently from \texttt{-O0} through \texttt{-O3}. These
examples illustrate global coherence and grounding in facts derived
from the binary; they are explanations of particular outcomes rather
than a separate evaluation claim. We compare with SK2Decompile, the
strongest neural refiner in our suite.

HE/132 illustrates a locally plausible but globally inconsistent type
choice. The reference prototype is \texttt{int func0(const char *)},
yet every baseline
emits a conflicting signature from a plausible \emph{local} inference:
IDA types the return \texttt{long long} from the 64-bit \texttt{RAX},
missing its later use as \texttt{int}; SK2Decompile and RetDec drop
the \texttt{const} from a parameter that is only read; Ghidra returns
\texttt{bool}, a C++ token in a C program. Each output fails to
compile or link against the reference. \decompiler{} carries the
candidates (\texttt{int}/\texttt{long},
\texttt{char*}/\texttt{const char*}) and lets one per-function
\textsc{MaxSMT} solve pick the assignment consistent across all
uses: \texttt{int} because a caller consumes it as such, and
\texttt{const} because no write occurs. The guarantee is coherence
across the function rather than per-site ground truth, and it is exactly what
site-local resolution cannot offer.

\section{Conclusion}

Decompilation is inference under irreducible ambiguity rather than
deterministic inversion. Classical pipelines commit too early to
interpretations whose validity depends on nonlocal evidence.
\framework{} separates two roles: monotone Datalog passes generate
candidates and preserve alternatives that remain globally coupled; a
per-function MaxSMT solver then commits jointly to the remaining
choices and produces a type-consistent C program. We implemented this
in \decompiler{} for stripped x86-64 ELF binaries.
It achieves a 93\% recompilation rate, whereas classical
baselines and a greedy ablation over the same candidate forest
recompile none. The remaining challenge is behavioral: only 7\% of the coreutils
test suite passes. Because generation is rule-defined, an
infeasible solve guarantees a missing rule rather than a failure in the
solver (Lemma~\ref{lem:unsat}). The behavioral gap can be closed by
adding rules, the focus of our future work.

\bibliography{ref}

\end{document}